\documentclass[a4paper]{article}
\usepackage[margin=3cm]{geometry}
\usepackage{latexsym}
\usepackage{amsmath}
\usepackage{amsthm}
\usepackage{amssymb}
\usepackage{amsfonts}
\usepackage[pdftex]{graphicx}

\newtheorem{theorem}{Theorem}
\newtheorem{remark}{Remark}

\newtheorem{conj}[theorem]{Conjecture}
\newtheorem{prop}[theorem]{Property}

\begin{document}

\title{Looking more closely to the Rabinovich-Fabrikant system}

\author{MARIUS-F. DANCA \\Dept. of Mathematics and Computer Science, \\Avram Iancu University of Cluj-Napoca, Romania\\and\\
Romanian Institute of Science and Technology, Cluj-Napoca, Romania\\danca@rist.ro\\
MICHAL FE\u{C}KAN \\Department of Mathematical Analysis and Numerical Mathematics\\Comenius University in Bratislava, Slovakia\\and\\
Mathematical Institute Slovak Academy of Sciences\\
Bratislava, Slovakia\\Michal.Feckan@fmph.uniba.sk\\
NIKOLAY KUZNETSOV \\Department of Applied Cybernetics\\Saint-Petersburg State University, Russia\\and\\
  University of Jyv\"{a}skyl\"{a}, Finland\\
  nkuznetsov239@gmail.com\\
GUANRONG CHEN \\Department of Electronic Engineering, \\City University of Hong Kong,\\ Hong Kong SAR, China\\eegchen@cityu.edu.hk}
\maketitle

\begin{abstract}Recently, we look more closely into the Rabinovich-Fabrikant system, after a decade of the study in \cite{danca1}, discovering some new characteristics such as cycling chaos, transient chaos, chaotic hidden attractors and a new kind of saddles-like attractor. In addition to extensive and accurate numerical analysis, on the assumptive existence of heteroclinic orbits, we provide a few of their approximations.
\end{abstract}

\emph{Rabinovich-Fabrikant system; cycling chaos; transient chaos; heteroclinic orbit; LIL numerical method}

\section{Introduction}

Rabinovich \& Fabrikant  [1979] introduced and analyzed from physical point of view a model describing the stochasticity arising from the modulation instability in a non-equilibrium dissipative medium. This is a simplification of a complex nonlinear parabolic equation modelling different physical systems, such as the Tollmien-Schlichting waves in hydrodynamic flows, wind waves on water, concentration waves during chemical reactions in a medium where diffusion occur, Langmuir waves in plasma, etc.

The mathematical model of Rabinovich \& Fabrikant [1979] is described by the following equations (named the RF system):

\begin{equation}
\label{rf}
\begin{array}{l}
\overset{.}{x}_{1}=x_{2}\left( x_{3}-1+x_{1}^{2}\right) +ax_{1}, \\
\overset{.}{x}_{2}=x_{1}\left( 3x_{3}+1-x_{1}^{2}\right) +ax_{2}, \\
\overset{.}{x}_{3}=-2x_{3}\left( b+x_{1}x_{2}\right),
\end{array}%
\end{equation}

\noindent where the two constant parameters $a,b>0$. For $a<b$, the system is dissipative:

\[
\textrm{div}(f(x))=\sum_{i=1}^{3}\frac{\partial }{\partial x_{i}}f_{i}(x)=2(a-b)<0.
\]

Roughly speaking, there are at least three reasons to reconsider this  RF system: one is the fact that the system models a physical system and, therefore, it is not an artificial model; another reason is the fact that, due to the strong nonlinearity, a rigorous mathematical analysis cannot be be performed on it, hence the system might reveal some new interesting characteristics; and, finally, it poses some real challenges to numerical methods for ODEs \cite{danca3}.

Compared with the numerical investigation reported in \cite{raba}, the studies in \cite{danca1} (and also in \cite{danca2}), revealed some new interesting aspects. Since then, the interest in this system, has continuously increased, partially following the direction of \cite{danca1} (see \cite{cit2,cit3,cit4,cit5,cit6,cit7,cit8,cit9}, utilizing some computer software such as Matlab Central \cite{cit10}  or Wolfram \cite{cit11}).

The system behavior depends sensibly on the parameter $b$ but not so much on $a$. As can be seen from the bifurcation diagram of the phase variable $x_3$ presented in Fig. \ref{bifurcation} (a), especially the zoomed detail $D$ (Fig. \ref{bifurcation} (b)), there are some hidden cascades of period doubling bifurcations, which do not appear using the aforementioned numerical software is used in a standard resolution.
With a single exception, we fix $a = 0.1$ and let $b$ be the bifurcation parameter. It is worth mentioning that even without some physical meaning, few interesting cases have been found with negative values of $a$ and $b$.

Because of the coexistence of chaotic attractors and stable cycles, and also because of the seemingly complicated attraction basins, obtaining simultaneously in the bifurcation diagram the evolutions of both stable and unstable equilibria is obviously a difficult task.

For some values of $b$, the results depend drastically on the step-size, the initial conditions, and the numerical methods used. So, the available efficient numerical methods for ODEs, implemented in different software packages, might give unexpectedly different results for the same parameter values and initial conditions, while fixed-step-size schemes (such as the standard Runge-Kutta method RK4, or the predictor-corrector LIL method (see Appendix and \cite{danca3}, utilized in this paper) generally give more accurate results, although in some cases these are strongly dependent on the step-sizes.

On the other hand, the attraction basins have an extremely complicated (fractal) boundaries, since for some given values of $b$ one can obtain several different attractors depending on the locations of the initial conditions, even if they suffer infinitesimally.

Because a complete mathematical analysis including the stability of the equilibria, the existence of invariant sets, the existence and convergence of heteroclinic or homoclinic orbits are impossible at this time, most investigations in the literature are based on numerical and computer-graphic analyses.

Following this common practice, this paper takes a numerical analysis-based approach, and in fact the numerical results in this paper are obtained by tedious trial-and-error.

For strongly chaotic systems, like the RF system, initial deviations from a true orbit can be magnified at a large exponential rate, making direct computational methods fail quickly \cite{jiz}. This feature could be responsible for some spectacular and interesting simulation results obtained for example by varying the step-size of the numerical method, or just by changing the numerical routine.

The  following numerical integrations and computer simulations have were obtained generally with the step-size $h=0.00005-0.0001$, while the integration time interval was $I=[0,T_{max}]$ with $T_{max}=300-500$. The initial conditions $S=(x_{0,1},x_{0,2},x_{0,3})$, have of a major impact on the numerical results, were chosen generally as follows: $x_{0,1}, x_{0,2}\in (-1,1)$, but mostly $x_{0,1}=-x_{0,2}=\pm0.05$, and $x_{03}=0.3$. Larger values for initial conditions could lead to system instability.

Except for some singular cases, all simulations were performed for $a=0.1$ and $b\in(b_{min},b_{max})$, with $b_{min}=0.13$ and $b_{max}=2$. However, it is noted that some interesting results were found also for $b\notin(b_{min} , b_{max} )$ and $a\neq 0.1$.

\vspace{3mm}
\emph{Some cases, are marked with ``*'' when some expected attractor presents an extreme dependence on initial conditions, on used numerical method, or step-size, and simulations are difficult.
}
\vspace{3mm}

A representative stable cycle, for $a=0.1$ and $b=1.035$, and a chaotic attractor, for $a=-1$ and $b=-0.1$, are presented in Fig. \ref{spictac}. The colored tubular representation shows the speeds along the attractors (red color and longer cylinders indicate the higher speeds).

In this paper, we are interested in the system saddles, their (inter)connections, chaotic behavior, coexisting attractors, cycling chaos, transient chaos, hidden attractors and on the virtual saddles-like attractor which are found numerically for the first time from the RF system.

The rest of the paper is organized as follows: Section 2 deals discusses the system equilibria, Section 3 presents numerical approximations of heteroclinic orbits connecting equilibria with stable cycles and chaotic attractors, and also the case of connecting two distinct chaotic attractors. Section 4  investigates different kind of attractors, such as coexisting chaotic attractors and stable equilibria, transient chaos and hidden attractors. Conclusion is summarized in the last section of the paper.

\section{System equilibria}

The system is equivariant with respect to the following symmetry:
\begin{equation}\label{sim}
T(x_1,x_2,x_3)\rightarrow(-x_1,x_2,x_3),
\end{equation}

This symmetry means that any orbit, not invariant under $T$, has its symmetrical (‘‘twin’’) orbit under this transformation $T$, namely, all orbits are symmetric one to another with respect to the $x_3$-axis. This symmetry persists for all values of the system parameters, which is also reflected in the expressions of the five equilibrium: $X_0^*(0,0,0)$ and the other four points:
\begin{equation*}
\begin{array}{l}
X_{1,2}^{\ast }\left( \mp \sqrt{\dfrac{bR_1+2b}{4b-3a}},\pm \sqrt{b\dfrac{4b-3a}{R_1+2}},\dfrac{aR_1+R_2}{\left(4b-3a\right) R_1+8b-6a}\right),
\\
X_{3,4}^{\ast }\left( \mp \sqrt{\dfrac{bR_1-2b}{3a-4b}},\pm \sqrt{b\dfrac{4b-3a}{2-R_1}},\dfrac{aR_1-R_2}{\left(
4b-3a\right) R_1-8b+6a}\right) ,
\end{array}%
\end{equation*}

\noindent where $R_1=\sqrt{3a^{2}-4ab+4}$ and $R_2=4ab^{2}-7a^{2}b+3a^{3}+2a$. In \cite{raba}, and also in \cite{danca1}, the equilibria are obtained by hand-computing, while in this paper, for computational reasons we use symbolic solvers. The Jacobian matrix is
\[
J=\left(
\begin{array}{ccc}
2x_{1}x_{2}+a & x_{1}^{2}+x_{3}-1 & x_{2} \\
-3x_{1}^{2}+3x_{3}+1 & a & 3x_{1} \\
-2x_{1}x_{3} & -2x_{1}x_{3} & -2\left( x_{1}x_{2}+b\right)
\end{array}%
\right).
\]

\subsection{$X_0^*$}

The equilibrium $X_0^*$ has the associated characteristic equation:
\begin{equation}
(\lambda^2-2a\lambda+a^2+1)(\lambda+2b)=0,
\end{equation}

\noindent with eigenvalues $\lambda_{1,2}=a\pm i$ and $\lambda_3=-2b<0$. Therefore, $X_0^*$ is a \emph{repelling focus saddle} (see \cite{thei}).

The stability of the other four points $X_{i}^*, ~ i=1,...,4$, cannot be evaluated in general by analytical means; therefore, a numerical approach with symbolic computation was used to calculate and analyze the eigenvalues. For all values of $b\in(b_{min},b_{max})$, the characteristic equations corresponding to equilibria $X_{1-4}^*$ have a pair of complex conjugate roots and a real root, denoted by $\lambda_{1,2}\in \mathbb{C}$ and $\lambda_3\in \mathbb{R}$ respectively. Clearly, all equilibria will determine scrolling dynamics, reflecting a complicate and spectacular aspect of system obtained attractors.

The eigenvalues plotted in the $b$-parameter space are graphically illustrated in Fig. \ref{eigenus} (a) where, because of the mentioned symmetry, $X_{1,2}^*$ and $X_{3,4}^*$ have same eigenvalues, respectively, which are further discussed below.

\subsection{$X_{1,2}^*$}

Equilibria $X_{1,2}^*$ have a negative real eigenvalue $\lambda_3$ for every $b\in(b_{min},b_{max})$ (circle in Fig. \ref{eigenus} (a)). On the other hand, there exists a tiny interval $(b_1,b_2)=(1.05,1.67)$, where the real parts of the complex roots $\lambda_{1,2}$ (diamond in Fig. \ref{eigenus} (a)) are positive: $Real(\lambda_{1,2})>0$ (the red portion in region $D$), while for $b\in(b_{min},b_1)\cup (b_2,b_{max})$, $Real(\lambda_{1,2})<0$. Therefore, $X_{1,2}^*$ is a \emph{stable focus node (sink)} for $b\in(b_{min},b_1)\cup (b_2,b_{max})$, where all orbits starting in close neighborhoods, or attraction basin of $X_{1,2}^*$, will be attracted from all direction to this equilibria, and is a \emph{repelling focus saddle} for $b\in(b_1,b_2)$. In this case, all orbits, starting in close neighborhoods or the attraction basin of $X_{1,2}^*$, will scroll inside towards $X_{1,2}^*$ on the surface determined by the 2-dimensional stable manifold and, after some finite time, determined by the presence of the imaginary roots, then will be pushed away in the direction of the straight outflow determined by the underlying eigenvector, toward either a stable cycle or a chaotic attractor.\footnote{These results improve the coarse ones stabilized in \cite{danca1}, where the study was only for $b\in(0.13,1.3)$ and therefore only $b_1$ could be obtained (with a lower accuracy at that time: $b_1=1.025$)}
\subsection{$X_{3,4}^*$}
Because $\lambda_{3}>0$ and $Real(\lambda_{1,2})<0$ for all $b\in(b_{min},b_{max})$ (Fig. \ref{eigenus} (b)), $X_{3,4}^*$ are \emph{attracting saddles} for all $b\in(b_{min},b_{max})$, so all orbits, starting from close neighborhoods or the attraction basin of $X_{3,4}^*$, will be attracted by $X_{3,4}^*$ on the surface determined by the 2-dimensional stable manifold and, at some moment of time, they exit along the direction of the 1-dimensional unstable manifold.

The results are centralized in Fig. \ref{tabel}, where the iconic representations indicate the stability type (see \cite{thei}).

\begin{remark}\label{remus2}
\setlength{\itemindent}{.1in}
\item [i)]
Despite the relative simple evolution of eigenvalues of $X_{3,4}^*$ in the $b$-parameter space, these equilibria are generally responsible for the system dynamics, including here the heteroclinic orbit: all found numerical approximations of heteroclinic orbits start from $X_{3,4}^*$;
\item [ii)]For particular cases of $X_{1,2}^*$ with $b\approx b_{1,2}$, when $Real(\lambda_{1,2})=0$ and $\lambda_3<0$ (Fig. \ref{eigenus} (a)) and when the hyperbolicity of $X_{1,2}^*$ vanishes, some bifurcations and rich dynamics are possible, but such situations are not considered in this paper.
\end{remark}

The influence of attraction/repulsion of the five equilibrium points gives the richness of the dynamics of the RF system as shown in Fig. \ref{bifurcation}.

\section{Numerical approximation of heteroclinic orbits}\label{hetero}

A heteroclinic orbit, $\Gamma$, between two equilibria $X^*$ and $Y^*$ of a dynamical system $\dot x=f(x)$ is a trajectory that is backward asymptotic to $X^*$ and forward asymptotic to $Y^*$. This means that the heteroclinic orbit $\Gamma(t)$, solution of the underlying initial value problem, must verify
\[\label{adevar}
\begin{array}{c}
\underset{t\rightarrow \infty }{\lim }\Gamma (t)=Y^{\ast }, ~\underset{t\rightarrow -\infty }{\lim }\Gamma (t)=X^{\ast }.
\end{array}%
\]

\vspace{3mm}
\emph{In this paper, by hyperbolic orbit it refers to a Numerical Approximated Heteroclinic Orbit (NAHO) in the phase space of a path starting from a close neighborhood of a saddle ($X_{3,4}^*$ here) and connects the saddle with another saddle ($X_{1,2}^*$), or a stable cycle, or chaotic attractor, or connects two chaotic attractors.
}\vspace{3mm}

The usual way to determine analytically a heteroclinic orbit for (\ref{rf}) is to solve, for example, $x_3$ from the first equation and then putting it into the 2nd and 3rd equations, so as to get a nonlinear higher order ODE for $x_1$ and $x_2$. Next, one can try to use the method presented e.g. in \cite{unu}, \cite{doi}, or \cite{fang} to expand $x_2$ and $x_3$ in exponential series and compute recurrent relations for undetermined coefficients, after which the uniform convergence of the series solution must be proven (see \cite{wig} for a comprehensive review of results regarding homoclinic and heteroclinic motions in three and four dimensions).

Another way is to realize that all the above unstable equilibria have a 1-dimensional manifold, which is either stable or unstable. Therefore, one consider that the heteroclinic solution coincides in this case with these manifolds. Note that these manifolds are graphs of curves tangent to the corresponding eigenvalues. Moreover, these curves can be expressed as power series in one variable, which can be rather effectively computed from the differential equation (\ref{rf}).

On the other side, there exist relatively new numerical algorithms to find heteroclinic orbits (see, for example, \cite{conect1}). Also, by assuming that between two saddles there exists a connection, for example between the saddles $X_{3,4}^*$ and $X_{1,2}^*$, this must be one of the trajectories contained in $W^u(X_{3,4}^*)\cap W^s(X_{1,2}^*)$. These kind of intersections can also be numerically determined \cite{thei}.

However, since for this system, taking account on his complexity, to prove analytically the exisitence of heteroclinic and (or) homoclinic orbits, or to use one of the existing algorithms to determine precisely these orbits, really is a practically tedious task.

Therefore, we will take a semi-analytic approach with extensive numerical simulation supports.

Specifically encouraged by the accurate computational results and by the fact that the existing symmetry is a natural setting for the existence of heteroclinic orbits, we motivated to propose the following conjecture regarding the heteroclinic orbits of the RF system (for simplicity, we do not consider the case of homoclinic orbits here).

\begin{conj}The RF system admits NAHOs.
\end{conj}

\noindent Regarding the equilibrium $X_0^*$, we can prove the following result.

\begin{prop}\label{th}
Equilibrium $X_0^*$ cannot have heteroclinic (homoclinic) orbits.
\end{prop}
\begin{proof}
By taking $x_3=0$ in \eqref{rf}, the planar reduced system is
 \begin{equation}
\label{rfred}
\begin{array}{l}
\overset{.}{x}_{1}=x_{2}\left(-1+x_{1}^{2}\right) +ax_{1}, \\
\overset{.}{x}_{2}=x_{1}\left(+1-x_{1}^{2}\right) +ax_{2}, \\
\end{array}%
\end{equation}

\noindent which satisfies

\begin{equation}\label{red}
\frac{d}{dt}(x_1^2+x_2^2)=2a(x_1^2+x_2^2).
\end{equation}

\noindent So, for $a>0$, $X_0^*$ is a global attractor for reduced system on $x_3=0$. Therefore, the origin $(0,0)$ is globally asymptotically unstable for (\ref{rfred}). Furthermore, the line $x_1=x_2=0$, i.e. the $x_3$-axis, is also invariant with the reduced equation

\begin{equation}
\dot{x_3}=-2bx_3,
\end{equation}

\noindent which has a solution $x_3(t)=e^{-2bt}x_2(0)$. Therefore, $X_0^*$ is attracting on $x_3$-axis and the origin $0$ of system (\ref{red}) is also globally asymptotically stable. This corresponds to the fact that $X_0$ is hyperbolic with a repelling saddle (having the 2-dimensional unstable $W_{X_0^*}^u$ and the 1-dimensional stable manifold $W_{X_0^*}^s$). Also, $W_{X_0^*}^u=\{x_3=0\}$ and $W_{X_0^*}^s=\{x_1=x_2=0\}$. This follows from the uniqueness of these invariant manifolds.

Next, if there would be a heteroclinic connection at $X_0^*$, then it would be lying in the stable or unstable manifolds of $X_0^*$. In the first case, it should be the $x_3$-axis, which is not a heteroclinic solution, however. In the second case, if it is the plane $x_3=0$, then it is an unbounded solution. So, either case, one has a contradiction. Consequently, $X_0^*$ has no heteroclinic connection.
\end{proof}

As mentioned before, NAHOs have been found by trial-and-error numerically, namely by searching adequate values for $b$ such that the orbits start as close as possible (in small neighborhoods) to $X_{3,4}^*$.

In the following, denote the coordinates of equilibria as $X_1^*(x_{i1}^*,x_{i2}^*,x_{i3}^*)$, $i=0,...,4$, and the NAHO as $\Gamma(x_1,x_2,x_3)$.

We present next the four main cases of numerical heteroclinic orbits approximations we found: $X_{3,4}^*\rightarrow X_{1,2}^*$, $X_{3,4}^*\rightarrow stable~ cycle$, $X_{3,4}^*\rightarrow chaotic~attractor$ and \emph{chaos} $\rightarrow$ \emph{chaos} and an interesting case of NAHOs connecting $X_{3,4}^*$ with two saddles-like. Due to the mentioned symmetries, we only consider the case of orbits starting from $X_3^*$, since the case of orbits starting from $X_4^*$ is similar.

\subsection{$b=0.288$}
The first NAHO was obtained for $b=0.288$, which connects $X_3^*$ to $X_1^*$ (Figs. \ref{asta-iatreia} (a), (b)). As the phase plots and their projections indicated (Figs. \ref{asta-iatreia} (c)-(e)), once the orbit enters by scrolling into a neighborhood of $X_3^*$, close to its 2-dimensional stable manifold, thereby and approaches by rotating around  $X_1^*$, because of the focus node type of stability of $X_1^*$ (see also Figs. \ref{asta-iatreia} (f)-(h)).
\subsection{$b=1.24$}
The second NAHO corresponds to $b=1.24$ (Figs. \ref{proba2} (a),(b)). Notice that although there is a connection between $X_0^*$ and $X_3^*$, denoted $\tilde{\Gamma}$, it is not a heteroclinic connection in virtue of Property \ref{th}. Another possible explanation, as why in this case the connection is not heteroclinic, is that the unstable manifolds $x_3$ is numerically so. In the third equation of (\ref{rf}), when $b+x_1x_2>0$, $x_3$-axis is attracting, but when $b+x_1x_2<0$ it is repelling.

Therefore, we have again a single NAHO connection between $X_3^*$ and a stable cycle $SC$. After some transients $T$, $SC$ is generated due to the lost of the stability of $X_3^*$ for this value of $b$ (Figs. \ref{proba2} (c)-(e)). Time series in Fig. \ref{proba2} (f)-(h) unveil this connection and also the stability of the cycle $SC$, while the histograms (Figs. \ref{proba2} (i)-(k)) indicate the $SC$'s periodic character.

Another interesting characteristic of the RF system is that the speed along its orbit varies significantly, especially along the $x_1$ component (see zoomed detail in the time series of $x_1$ in Fig. \ref{proba2} (f), and the peek in the underlying histogram in Fig. \ref{proba2} (i)). Thus, when the component $x_1$ of the orbit, $\Gamma(x_1)$, joints the neighborhood of the component $x_{11}^*$ of $X_1^*$, it remains for a relatively long time in this neighborhood. The speed along the $x_1$-axis, when the system orbit approaches the component $x_{11}^*$, is very small, fact revealed by zoomed detail in Fig. \ref{proba2} (f) and the related histogram. Also, the strong oscillations of the $x_3$ component are underlined by its time series (Fig. \ref{proba2} (h)), the related histogram (Fig. \ref{proba2} (k)), and also the tubular representation in  Fig. \ref{spictac} (a), where the varying speed on a stable cycle is indicated by colors. There, for this heteroclinic orbit, the orbit speeds in neighborhood of the equilibria $X_{1,2}^*$ are higher.

In this case the approximation is coarser since the NAHO exits in a larger neighborhood of $X_3^*$.

\subsection{$b=1.2128$}
For $b=1.2128$, there exists a connection between $X_3^*$ and, this time, a chaotic attractor born from the previous stable cycle $SC$ which, via a cascade of bifurcations (see the bifurcation diagram in Fig. \ref{bifurcation}) lost its stability (Fig. \ref{proba3}). As in the previous case, the orbit connecting $X_0^*$ should not be considered a NAHO. Again, the time wasted by the component $x_1$ of $\Gamma$, $\Gamma(x_1)$, in the neighborhood of $x_{11}^*$, is longer and can be identified from the time series (Fig. \ref{proba3} (f)) and the histogram (Fig. \ref{proba3} (i)).


\subsection{$b=1.2128$ (different initial conditions)}
Fig. \ref{hetero_x} presents another NAHO connection ($\Gamma_1$ and $\Gamma_2$) which now, from different initial conditions than in the above case, link two chaotic attractors. This case is interesting since it resembles the \emph{cycling chaos} (see e.g. \cite{cycl1,cycl2}. As is well known, saddle connections between equilibria can appear in systems with symmetries, and these connections can be cycling like in this case. Therefore, in this case the system is said to have a NAHO cycle).

\subsection{$b=1.8$}

Fig. \ref{crazy_bun} presents one of the most interesting cases for this system, which reveals two new scrolls, denoted by $Y_{1,2}^*$. Since there are only five equilibria, $Y_{1,2}^*$ cannot be equilibria but only reflection-like equilibria $X_{3,4}^*$. Therefore, taking into account the shape of the trajectories around, one may consider $Y_{1,2}^*$ as a kind of \emph{virtual saddles}.

The connections (curves $\tilde{\Gamma}_{1,2}$) between $X_{3,4}$ and $Y_{1,2}^*$ could be considered as NAHOs.

Notice that $Y_{1,2}^*$ could be clear obtained with the \emph{LIL} method (see \cite{danca3}), which yielded the most accurate numerical result (Fig. \ref{crazy_bun} (a)), while with the RK4 method, a less-accurate result was obtained (Fig. \ref{crazy_bun} (b)). The ode23 (Matlab solver) gave only the paths to $Y_{1,2}^*$, after which, for a longer time of integration, the trajectories diverge (Fig. \ref{crazy_bun} (c)).

As can be seen, while the scrolls around $X_3^*$ and $X_4^*$ are contrary to each other, the scrolls around $X_3^*$ and $Y_1^*$ and around $X_4^*$ and $Y_2^*$, respectively, take place in the same manner.

Another characteristic of $Y_{1,2}^*$ is that the system is unstable after a relatively short time ($T_{max}\approx 67$), and this could be one of reasons for other numerical integration routines to fail.


\begin{remark}
Possible connections $X_{1,2} \rightarrow X_{3,4}$ have not been found.
\end{remark}

\section{Chaotic attractors}\label{coe}

In this section, we present different kinds of chaotic attractors of the RF system.

\subsection{Chaotic attractors}
\noindent The system has several different chaotic attractors with different shapes (Fig. \ref{haosuri}). Also, having five equilibria, the RF system is topologically non-equivalent to many classical systems, such as the Lorenz and Chen systems (with three equilibria), R\"{o}ssler system (with two equilibria), some Sprott systems (with one equilibrium \cite{spr}), and so on.

Notice that, in general, the existence of chaos not necessarily implies the existence of heteroclinic orbits. However, in this case, it seems that this is possible (see the NAHO in Fig. \ref{hetero_x}).

As the bifurcation diagram and the zoomed detail $D$ indicates (Fig. \ref{bifurcation} (a) and (b)), there are intervals of $b$ with which chaos is possible to exist (Figs. \ref{haosuri} (a)-(f)).

While the first four chaotic attractors (Figs. \ref{haosuri} (a)-(d)) were obtained for $b\in(b_{min},b_{max})$ and $a=0.1$, the last chaotic attractor in Fig. \ref{haosuri} (e) was obtained for a no-physical meaning set of parameter values: $a=-1$ and $b=-0.1$.

\begin{remark}
\setlength{\itemindent}{.1in}
\item
The scrolling in the regions close to equilibria $X_{1,2}^*$, for all studied cases of either regular motion or chaotic motion, is the same as indicated in Fig. \ref{noua}.
\item Considering that there exist heteroclinic orbits, as \v{S}i`lnikov's Theorem requires \cite{sil1,sil2}, and by applying numrically \v{S}i`lnikov's criterion for $b\in(b_{min},b_{max})$, we predict that the RF system would have Smale horse-shoe type chaos for $b\in(0.13,0.199)$. On the other hand, the bifurcation diagram and numerically tests indicate that there exist no chaotic motion and NAHOs in this region of parameter $b$.\footnote{This fact underlines the importance of assumptive existence of heteroclinic connections required by Silnikov's theorem, to prove chaotic motion in the sense of Smale horse-shoe.}
\end{remark}

\subsection{Transient chaos }
Another interesting observation of an apparent chaotic behavior, appearing when $b=0.279$, as indicated by the bifurcation diagram (see Fig. \ref{bifurcation} (b)), should enhance the chaotic motion. However, as the phase plot in Fig. \ref{interesant} shows, the transients are considerably long and, therefore, one can consider it as \emph{transient chaos} (see e.g. \cite{lai}). Finally, the system seems to have ``self-control'', thereby destroying the chaotic behavior. The phenomenon in this case seems to be linked to the less visible chaotic window starting at $b=b^*$ ($D_1$ in Fig. \ref{bifurcation} (b)), to which $b=0.279$ belongs. It at first sight indicates the coexistence of a chaotic attractor and the stable focus node $X_{1,2}$ which, after a sufficiently long time, is destroyed by the stronger influence of the stability of $X_{1,2}$. This situation underlines our presumed complexity of the inter woven fractal basin boundary structure of the RF system.
 \subsection{Coexisting attractors}

 \noindent Several equilibrium states (attractors) may coexist for a given set of system parameters \cite{four,spr2}. This phenomenon, one of most exciting in nonlinear dynamics, is referred to \emph{multistability} and has been found in almost all research areas of natural science, such as mechanics, optics, electronics, environmental science and neuroscience. Multistable systems are characterized by a high degree of complexity in behavior due to the “interaction” among different attractors \cite{multi}. In these cases, the qualitative behavior of the system might change dramatically under the variation of the system parameters, as in the RF system.

Fig. \ref{coexist} presents the following coexistences: \emph{chaotic~ attractors\textemdash sinks }$X_{1,2}^*$ (Figs. \ref{coexist} (a)-(c) for $b=0.277,~b=0.2876,~ b=0.98$, respectively), \emph{stable cycle\textemdash stable cycle} (Fig. \ref{coexist} (d) for $b=1.08$) and \emph{stable cycles\textemdash attractive point} (Fig. \ref{coexist} (e) for $b=1.035$).

\vspace{3mm}
\emph{To note that we did not found coexisting stable cycle\textemdash chaotic attractors.}
\vspace{3mm}

\subsection{Hidden attractors}
As defined in \cite{fish}, an attractor is called a \emph{hidden attractor }if its basin of attraction does not intersect with small neighborhoods of equilibria (see also \cite{kuzne2,kuzne} for details about hidden attractors). Therefore, to find hidden attractors of the RF system, we have to check numerically by choosing the initial points on the unstable manifolds, in small vicinity of the equilibria, and integrating the system whether we can see the obtained trajectories are attracted to the chaotic attractor.

\indent Here, consider the case of $a=0.1, b=0.2715$, when there is a chaotic attractor besides the stable equilibria $X_{1,2}^*$ (see Fig. \ref{tabel}).

In Fig. \ref{RF-3D}, the chaotic attractor (black) does not attract the two-dimensional unstable manifolds $W_{X_0^*}^u$ of $X_0^*$, since, as can be seen in the figure, the planar curves, lying in $x_3=0$ (see the stability of $X_0^*$ discussed in Section 2), with initial conditions in $W_{X_0^*}^u$ (blue) ``grows'' as the time of integration increases which does not intersect the chaotic attractor basin of attraction. Also, the separatrices of $X_{3,4}^*$ (red) turn to infinity and also do not intersect the attractor basin of attraction. This numerical analysis gives us a very good reason to say (but very carefully, taking into account all the difficulties arising from the numerical investigation of this system) that the chaotic attractors obtained in system (\ref{rf}) may be hidden. Several other cases (such as $b = 0.2876$, $b = 0.98$) gave similar results (see also Figs. \ref{coexist} (a)-(c) and Fig. \ref{interesant}).

As mentioned in \cite{kuzne}, the existence of hidden attractors, in the present RF could be a consequence of attractors coexistence or system multistability.


\section*{Conclusion}

In this paper, we have revisited the Rabinovich-Fabrikant system via careful and accurate numerical analysis, to unveil some new and interesting dynamical properties and behaviors of the system. In addition to the improvements of some previous results presented in \cite{danca1}, we found new numerical approximations of heteroclinic orbits, under the assumption that these orbits exist. Beside the coexistence between several types of attractors, cycling chaos, hidden attractors, transient chaos, we also found numerically that this system could present a different kind of saddle-like, which can be unveiled only by fixed step-size predictor-corrector LIL method (Appendix). Both the RK4 method and LIL method were utilized.
In the future, beyond the semi-analytic or numerical approaches taken in this paper, which seems to be the only option today, it would be nice and useful to have more rigorous analytic methodologies and approaches for this kind of investigations.

\section*{Appendix}

\noindent \textbf{Multistep Predictor Corrector Local Iterative Linearization (LIL) Method}

\noindent Consider the following initial value problem

\begin{equation}
\dot{x}=f(t,x),~~ x(t_0)=x_0,
\end{equation}

\noindent where $f:[t_0,T]$ with $T>0$, $t_0\in \mathbb{R}_+$, is a $\mathbb{C}^m$-smooth Lipschitz function.

The $m$-step ($m=3$ this paper) predictor-corrector Local Iterative Linearization (LIL) method is defined as

\begin{equation}\label{lil}
x_k=\frac{5}{3}x_{k-1}-\frac{13}{15}x_{k-2}+\frac{1}{5}x_{k-3}+\frac{h}{45}[26f_k-5f_{k-1}+4f_{k-2}-f_{k-3}],
\end{equation}

\noindent where $f_k=f(t_k, x_k)$ for $k=0,2,...$. Since the method is implicit, the corrector form (\ref{lil}), requires a predictor determination of $x_k$ (appearing in $f_k$). By the LIL method

\begin{equation}
x_{k}=3x_{k-3}-3x_{k-2}+x_{k-1}.
\end{equation}

Also, as a predictor corrector formula, the LIL scheme requires a starting fixed-$m$-step-size method to generate the necessary initial steps ($x_{-1}, x_{-2},$ and $x_{-3}$ for $m=3$). In this paper, the first three steps were generated by the RK4 method, which was used before LIL is started.

The convergency, time stability, comparison with other standard methods, formulas for several $m$ and some applications of the $LIL$ scheme can be found in \cite{danca3}.

\vspace{3mm}

\noindent \textbf{Acknowledgment}
MF is supported by the Grants VEGA-MS 1/0071/14, VEGA-SAV 2/0029/13 and by the Slovak Research and Development Agency under the contract  No. APVV-14-0378. NK thanks Russian Scientific Foundation (project 14-21-00041).  GC appreciates the GRF Grant CityU11208515 by the Hong Kong Research Grants Council.

\newpage{\pagestyle{empty}\cleardoublepage}

\pagebreak

\begin{figure}[t!]
\begin{center}
\includegraphics[clip,width=0.85\textwidth]{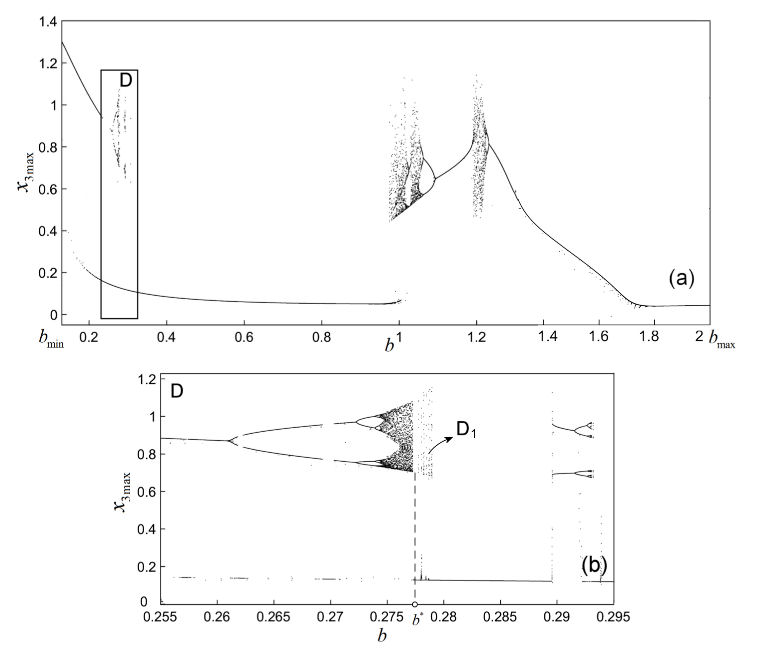}
\caption{a) Bifurcation diagram of $x_3$ of the RF system (local maxima are plotted). (b) Zoomed (rescaled) detail $D$.}
\label{bifurcation}
\end{center}
\end{figure}

\begin{figure}[t!]
\begin{center}
\includegraphics[clip,width=0.8\textwidth]{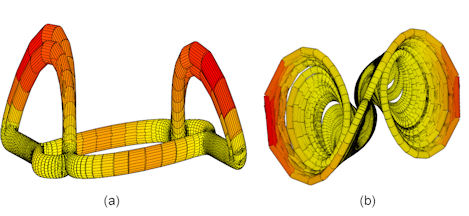}
\caption{Tubular representations of: (a) Stable cycle obtained with $a=0.1$ and $b=1.035$. (b) Chaotic attractor obtained with $a=-1$ and $b=-0.1$.}
\label{spictac}
\end{center}
\end{figure}

\begin{figure}[t!]
\begin{center}
\includegraphics[clip,width=0.9\textwidth]{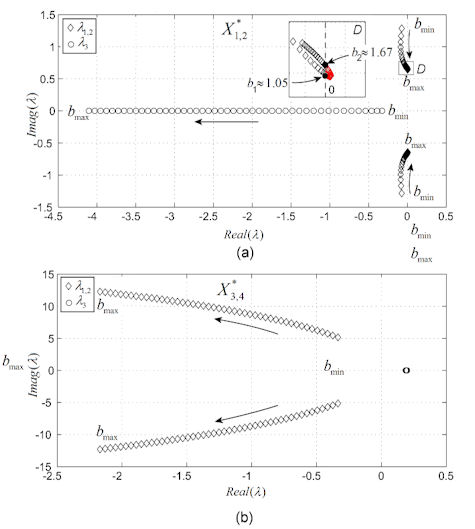}
\caption{Eigenvalues of the RF system in the $b$-parameter space. (a) Eigenvalues of equilibria $X_{1,2}^*$. (b) Eigenvalues of equilibria $X_{3,4}^*$.}
\label{eigenus}
\end{center}
\end{figure}

\begin{figure}[t!]
\begin{center}
\includegraphics[clip,width=0.75 \textwidth]{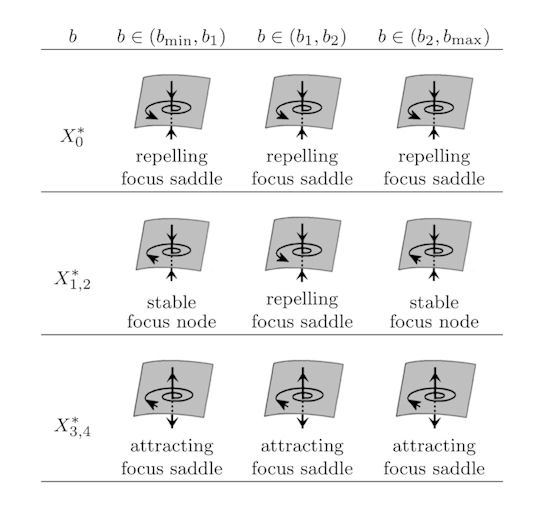}
\caption{Equilibria of the RF system (\ref{rf}). $b_{min}=0.13$, $b_1=1.05$, $b_2=1.67$ and $b_{max=2}$.}
\label{tabel}
\end{center}
\end{figure}

\begin{figure}[t!]
\begin{center}
\includegraphics[clip,width=1\textwidth]{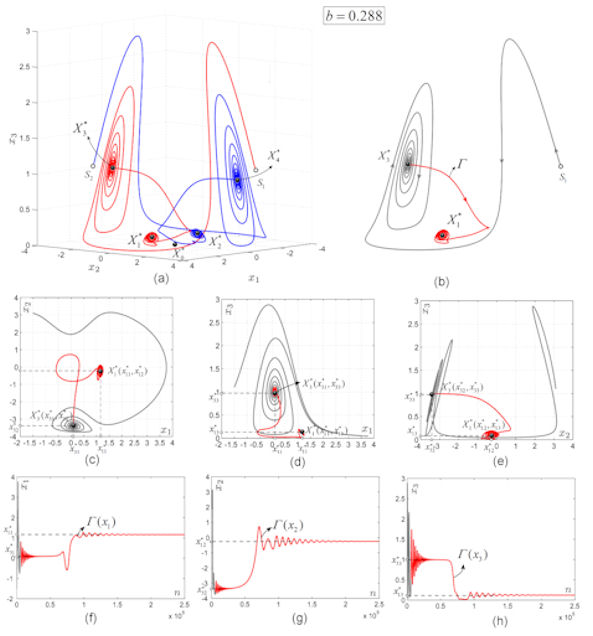}
\caption{NAHO for $b=0.288$. (a) Phase plots of the two symmetrical NAHOs. (b) NAHO connecting $X^*_3$ with $X_1^*$. (c)-(d) Plane projections. (f)-(h) Time series.}
\label{asta-iatreia}
\end{center}
\end{figure}

\begin{figure}[t!]
\begin{center}
\includegraphics[clip,width=0.8\textwidth]{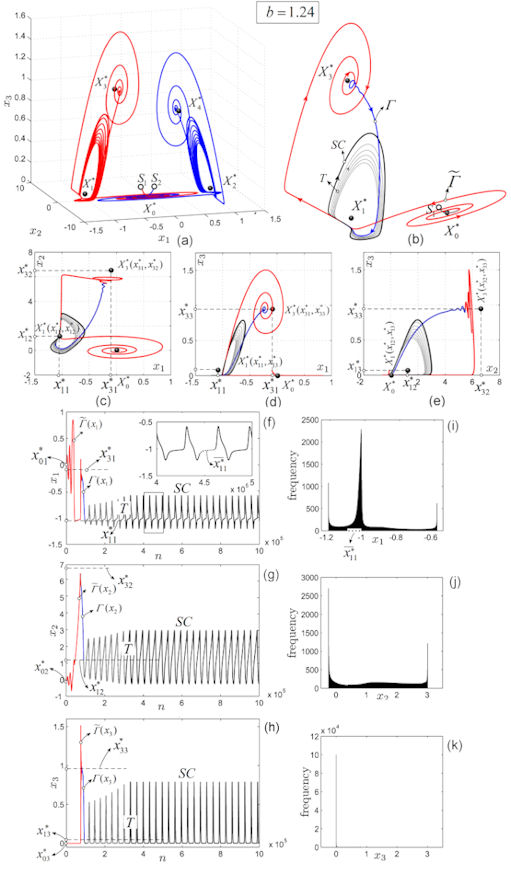}
\caption{NAHO for $b=1.24$. (a) Phase plots of the two symmetrical NAHOs. (b) NAHO connecting $X^*_3$ with the stable cycle $SC$. (c)-(d) Plane projections. (f)-(h) Time series. (i)-(k) Histograms.}
\label{proba2}
\end{center}
\end{figure}

\begin{figure}[t!]
\begin{center}
\includegraphics[clip,width=0.8\textwidth]{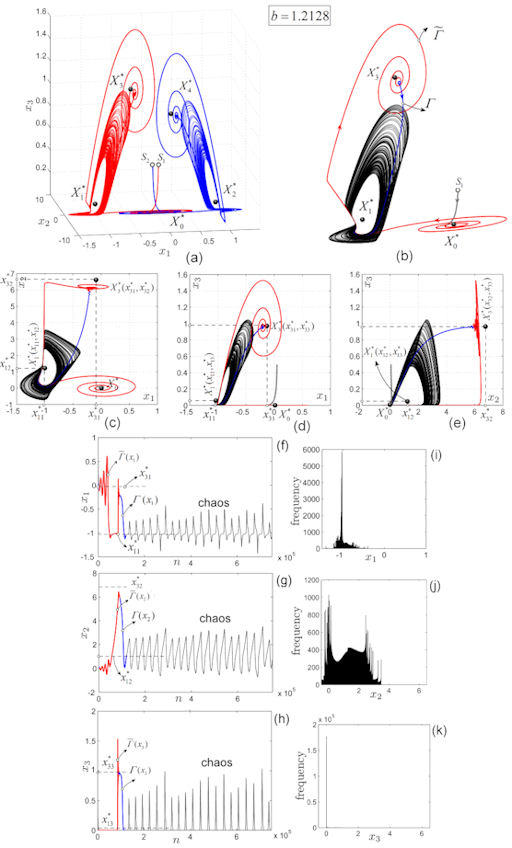}
\caption{NAHO for $b=1.2128$. (a) Phase plots of the two symmetrical NAHOs. (b) NAHO connecting $X^*_3$ to a chaotic attractor. (c)-(d) Plane projections. (f)-(h) Time series. (i)-(k) Histograms.}
\label{proba3}
\end{center}
\end{figure}

\begin{figure}[t!]
\begin{center}
\includegraphics[clip,width=0.9\textwidth]{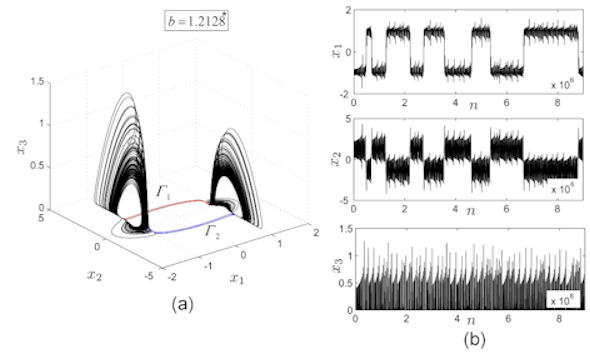}
\caption{Two NAHOs (cycling chaos) for $b=1.2128$, but with different initial conditions compared to the case in Fig. \ref{proba3}. (a) Phase plot. (b) Time series.}
\label{hetero_x}
\end{center}
\end{figure}

\begin{figure}[t!]
\begin{center}
\includegraphics[clip,width=0.95 \textwidth]{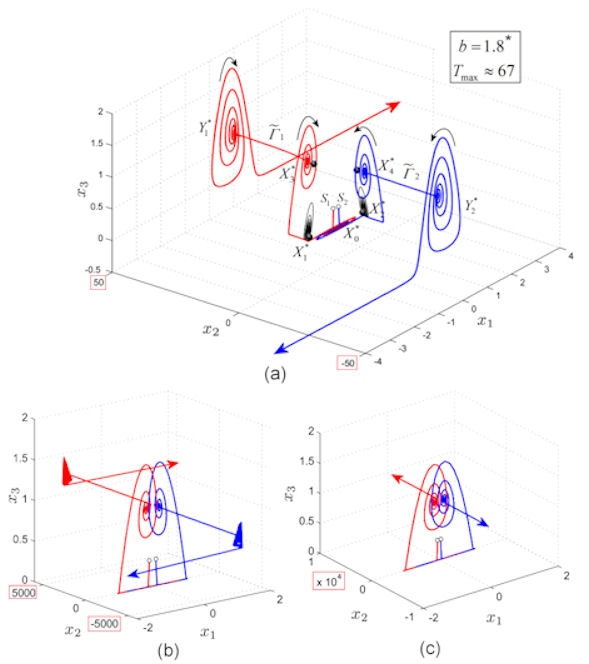}
\caption{Virtual saddles $Y_{1,2}^*$. (a) $Y_{1,2}^*$ obtained with the LIL scheme. (b) $Y_{1,2}^*$ obtained with the RK4 method. (c) $Y_{1,2}^*$ obtained with the ode23 Matlab solver (the arrows indicate the divergence).}
\label{crazy_bun}
\end{center}
\end{figure}

\begin{figure}[t!]
\begin{center}
\includegraphics[clip,width=0.9\textwidth]{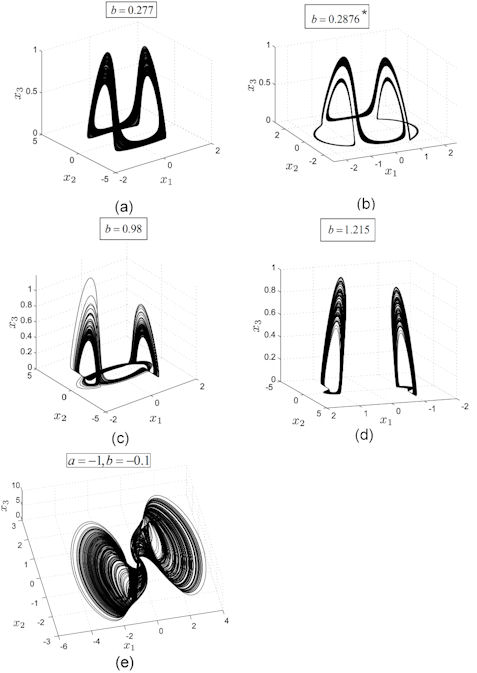}
\caption{Different chaotic attractors. (a) $b=0.277$. (b) $b=0.2876$. (c) $b=0.98$. (d) $b=1.215$. (e) $a=-1,~ b=-0.1$.}
\label{haosuri}
\end{center}
\end{figure}

\begin{figure}[t!]
\begin{center}
\includegraphics[clip,width=0.9\textwidth]{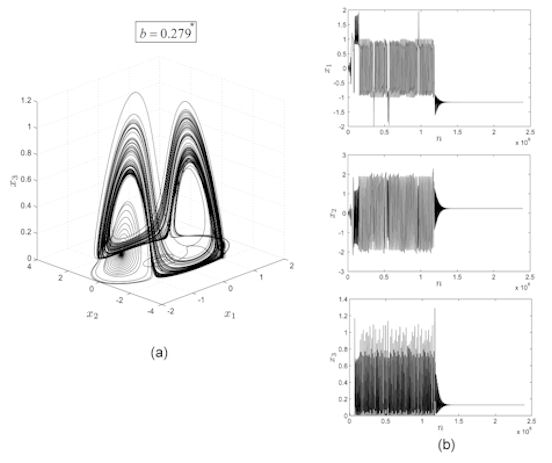}
\caption{Transient chaos for $b=0.279$.}
\label{interesant}
\end{center}
\end{figure}

\begin{figure}[t!]
\begin{center}
\includegraphics[clip,width=0.6\textwidth]{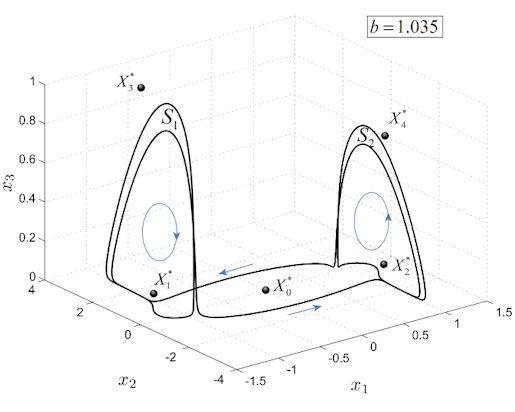}
\caption{A typical stable cycle of the RF system for $a=0.1$ and $b=1.035$. }
\label{noua}
\end{center}
\end{figure}

\begin{figure}[t!]
\begin{center}
\includegraphics[clip,width=0.8\textwidth]{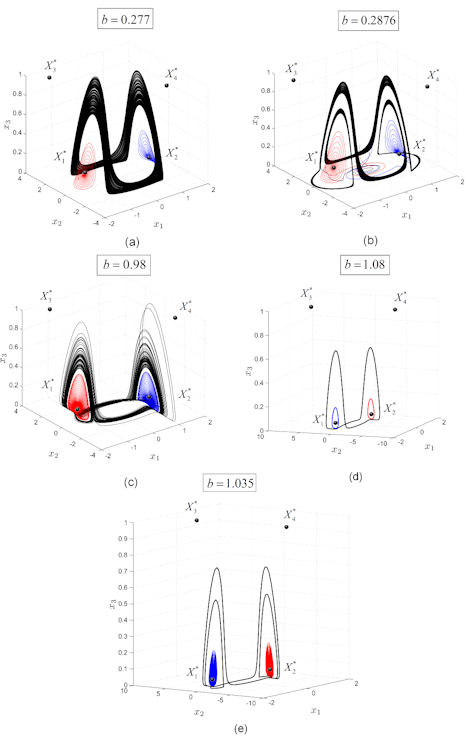}
\caption{Coexisting attractors. (a)-(c) Chaotic attractors with the two stable fixed points $X_{1,2}^*$ for $b=0.277$, $b=0.2876$ and $b=0.98$ respectively. (d) Coexisting three stable cycles for $b=1.08$. (e) Coexisting stable cycle with the two stable fixed points $X_{1,2}^*$ for $b=1.035$.}
\label{coexist}
\end{center}
\end{figure}

\begin{figure}[t!]
\begin{center}
\includegraphics[clip,width=0.75 \textwidth]{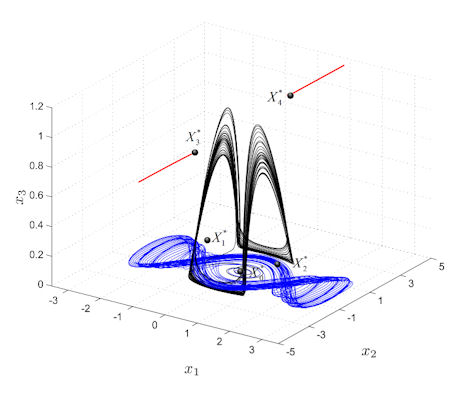}
\caption{A possible hidden attractor for $a=0.1$ and $b=0.2715$. Black: possible hidden attractor; Red: separatrices of $X_{3,4}^*$; Blue: planar trajectory with initial condition on the two-dimensional unstable manifold $W_{X_0^*}^u$ of $X_0^*$.}
\label{RF-3D}
\end{center}
\end{figure}

\end{document}